\documentclass[letterpaper, 10pt, conference]{ieeeconf} 
\IEEEoverridecommandlockouts                           

\usepackage{graphics} 
\usepackage{epsfig} 
\usepackage{mathptmx}
\usepackage{times} 
\usepackage{amsmath}
\usepackage{amssymb}
\usepackage{dsfont}
\usepackage{enumerate}
\usepackage[tight]{subfigure}
\usepackage{multirow}
\usepackage{bm}
\usepackage{mathtools}

\usepackage{color}

\newcommand{\K}{K}

\title{
On the influence of noise in randomized consensus algorithms*
}

\author{Renato Vizuete, Paolo Frasca, and Elena Panteley
\thanks{*Research supported in part by the French National Science Foundation
(ANR) via grant “Hybrid And Networked Dynamical sYstems” (HANDY), number ANR-18-CE40-0010. 
}
\thanks{R.~Vizuete and E.~Panteley are with Universit\'{e} Paris-Saclay, CNRS, CentraleSup\'{e}lec, Laboratoire des signaux et syst\`{e}mes, 91190, Gif-sur-Yvette, France. R.~Vizuete and P.~Frasca are with Univ.\ Grenoble Alpes, CNRS, Inria, Grenoble INP, GIPSA-lab, F-38000 Grenoble, France.
{\tt\small renato.vizuete@gipsa-lab.fr},
{\tt\small paolo.frasca@gipsa-lab.fr},\protect\\
{\tt\small elena.panteley@l2s.centralesupelec.fr}.}
}

\newcommand{\vertiii}[1]{{\left\vert\kern-0.25ex\left\vert\kern-0.25ex\left\vert #1 
    \right\vert\kern-0.25ex\right\vert\kern-0.25ex\right\vert}}
    
\newcommand{\nones}{N^{-1}\mathds{1}\mathds{1}^T}
\newcommand{\vect}{\mathrm{vec}}

\newtheorem{theorem}{Theorem}

\newtheorem{prop}{Proposition}
\newtheorem{lemma}{Lemma}
\newtheorem{remark}{Remark}
\newtheorem{assumption}{Assumption}

\begin{document}

\maketitle
\thispagestyle{empty}

\begin{abstract}
In this paper we study the influence of additive noise in randomized consensus algorithms. Assuming that the update matrices are symmetric, we derive a closed form expression for the mean square error induced by the noise, together with upper and lower bounds that are simpler to evaluate. Motivated by the study of Open Multi-Agent Systems, we concentrate on Randomly Induced Discretized Laplacians, a family of update matrices that are generated by sampling subgraphs of a large undirected graph. For these matrices, we express the bounds by using the eigenvalues of the Laplacian matrix of the underlying graph or the graph's average effective resistance, thereby proving their tightness. Finally, we derive expressions for the bounds on some examples of graphs and numerically evaluate them.
\end{abstract}

\section{Introduction}

Randomized interactions have been extensively studied in consensus systems~\cite{tahbaz2009consensus,ravazzi2015ergodic}, with a wide range of applications including social networks \cite{acemouglu2013opinion},
sensor networks \cite{kar2008sensor}, 
and clock synchronization \cite{bolognani2015randomized}. 
Even if consensus in random networks has been studied for long time, 
researchers have mostly considered ideal, noiseless, conditions for the interactions \cite{fagnani2008randomized}.  
Instead, more realistic dynamical models should at least include noise. 
The influence of noise has indeed been investigated in several works \cite{xiao2007distributed,jadbabaie2018scaling,yaziciouglu2019structural} regarding deterministic consensus, both in discrete-time \cite{garin2010survey} and in continuous-time \cite{patterson2010leader,bamieh2012coherence}.
Instead, relatively little is known about the effects of noise on randomized consensus. 
For instance, the authors in \cite{wang2012distributed} studied additive noise in random consensus over directed graphs, showing that, due to asymmetric loss of links, uncommon phenomena such as L\'{e}vy flights can arise in the evolution of the system. However, the emergence of this phenomenon is not possible for undirected graphs, provided the links are deactivated in a symmetric way. Nevertheless, noise prevents the states of the nodes from reaching consensus.

In this paper, we quantify the effects of noise by a noise index, which is simply the steady-state normalized mean square error between the states of the network and the consensus, and we study this index for symmetric update matrices.
In our first result, we obtain an explicit expression for it. Since this closed form involves the calculation of an $N^2$-dimensional matrix, we derive an upper bound and a lower bound that depend on the eigenvalues of relevant $N$-dimensional matrices. These results generalize well-known results about deterministic consensus with noise
~\cite{garin2010survey,bamieh2012coherence}. 

Next, we refine our analysis for a specific class of random update matrices, which we refer to as Randomly Induced Discretized Laplacians (RIDL). Given an underlying large graph, these matrices are generated by sampling a subset of active nodes and considering the subgraph induced by the active nodes. 
This way of sampling random update matrices, which to best of our knowledge has not been studied before in the literature, is motivated by the study of Open Multi-Agent Systems (OMAS): that is, multi-agent systems where agents can leave and join at any time \cite{hendrickx2016open,hendrickx2017open,franceschelli2018proportional,varma2018open}. 
When the update matrices are RIDLs, our bounds can be expressed as functions of the Laplacian eigenvalues of the underlying graph. Further rewriting of the bounds as functions of the average effective resistance reveals that they are asymptotically tight.

This paper is organized as follows. In Section II, we present the noise index, its closed form expression, and upper and lower bounds for its calculation. In Section III, we define RIDL matrices and illustrate their connection with OMAS. Then, we provide expressions for the bounds based on the spectrum of expected matrices. Finally, several examples of graphs are presented in Section IV and Section~V concludes the work.

\section{Noise in randomized consensus}

We consider the interactions between agents through a randomized averaging algorithm given by:
\begin{equation}\label{eq:dynamics}
x(t+1)=P(t)x(t),
\end{equation}
where $x(t)=[x_1(t),\ldots,x_N(t)]^T$ is a vector with the states of the nodes of the network and $P(t)$ is a stochastic matrix. 
One of the most important results on dynamics \eqref{eq:dynamics} is the determination of sufficient conditions to achieve probabilistic consensus with independent and identically distributed (i.i.d.) matrices $P(t)$.  
\begin{lemma}[Corollary 3.2 \cite{fagnani2008randomized}]\label{lemma:conditions}
The dynamics \eqref{eq:dynamics}  with i.i.d. matrices $P(t)$ reaches consensus
with probability 1
if:
\begin{subequations}
\begin{align}
P_{ii}(t)>0 \text{ with probability 1 for all }  i\in V , \label{eq:condi1} \\
G_{\bar P} \text{ has a globally reachable node}, \label{eq:condi2}
\end{align}
\end{subequations}
where $V=\{1,\ldots,N\}$ and $G_{\bar P}$ is the graph associated\footnote{Graph $G_{\bar P}$ is formed by adding edges between two nodes $i,j$ in $V$ when $\bar P_{ij}>0$, for all $i,j$, including possible self-loops.} with the expected matrix $\bar P=\mathbb{E}[P(t)]$. 
\end{lemma}

Condition \eqref{eq:condi1} is easy to satisfy in most cases since each agent usually has access to his own state. For undirected graphs, condition \eqref{eq:condi2} is equivalent to having a connected network. Moreover, for undirected graphs, the matrices $P(t)$ of dynamics \eqref{eq:dynamics} are usually symmetric and, therefore, doubly stochastic: if matrices $P(t)$ are doubly stochastic, then dynamics~\eqref{eq:dynamics}  
reaches average consensus with probability 1. 

Even if under ideal conditions it is possible to reach average consensus for undirected graphs, noise can affect the performance of the system. In order to study its effects, additive noise can be included in the dynamics \cite{garin2010survey,fagnani2017introduction}, by defining: 
\begin{equation}\label{eq:dynamics_noise}
x(t+1)=P(t)x(t)+n(t),
\end{equation}
where $n(t)$ is a vector of noise. It is natural to define a performance index to measure the deviation from the consensus point due to the effect of the perturbations as:
\begin{equation}\label{eq:noise_index}
J_{\mathrm{noise}}:=\dfrac{1}{N}\lim_{t\rightarrow{+\infty}}\mathbb{E}\left[\left\Vert x(t)-N^{-1}\mathds{1}\mathds{1}^Tx(t)\right\Vert^2\right],
\end{equation}
where $\mathds{1}$ is a vector of ones. The purpose of this paper is studying this index, under the following standing assumption.
\begin{assumption} The stochastic matrices $P(t)$ are i.i.d., symmetric, and satisfy \eqref{eq:condi1}-\eqref{eq:condi2}.
\end{assumption}
Our first contribution in this direction is showing that, under suitable conditions on the noise vector $n(t)$, index $J_{\mathrm{noise}}$ can be expressed in closed form.
\begin{theorem}[Mean square error index]
Consider system \eqref{eq:dynamics_noise} under Assumption~1 with a noise vector $n(t)$ uncorrelated w.r.t. both $i$ and $t$, with zero mean and variance $\sigma^2$. Then, the noise index \eqref{eq:noise_index} can be expressed as:
\begin{equation}\label{eq:jnoise}
J_{\mathrm{noise}}=\dfrac{\sigma^2}{N}\left(\vect^T(I_N)(I_{N^2}-\K)^{-1}\vect(I_N)-1\right),
\end{equation}
where $\K=\mathbb{E}[P(t)\Omega\otimes P(t)]$, $\Omega=I_N-\nones$ and $I_m$ is the identity matrix of size $m$.
\end{theorem}

\begin{proof}
For every $t$ we have:
$$
x(t)=Q(t)x(0)+\sum_{s=0}^{t-1}Q(s)n(t-s-1),
$$
where $Q(t)=P(t-1)\ldots P(1)P(0)$.
Let us denote $H=\nones$ and $\delta(t)=x(t)-\nones x(t)=x(t)-Hx(t)$. Then,
$$
\delta(t)=(Q(t)-H)x(0)+\sum_{s=0}^{t-1}(Q(s)-H)n(t-1-s).
$$
The expectation of its squared norm can be expressed as:
\begin{align}
&\mathbb{E}[\Vert\delta(t)\Vert^2]
=\mathbb{E}\left[\Vert (Q(t)-H)x(0)\Vert^2\right]\label{eq:3terms} \\\nonumber
&+2\sum_{s=0}^{t-1}x^T(0)(Q(t)-H)^T(Q(s)-H)\mathbb{E}[n(t-1-s)] \\\nonumber
&+\!\sum_{s,r=0}^{t-1}\!\mathbb{E}\left[[(Q(s)-H)n(t-s-1)]^T[(Q(r)-H)n(t-r-1)]\right]. 
\end{align}
First, we analyse the third term of \eqref{eq:3terms} which we denote by $U$. By applying the trace on the scalar $U$, we obtain 
$$
U\! =\!\sum_{s,r=0}^{t-1}\!\!\mathrm{tr}\left(\mathbb{E}\left[(Q^T\!\! (s)Q(r)-H)\right]\!\mathbb{E}[n(t-1-s)n^T\!\! (t-1-r)]\right).
$$
Due to the characteristics of the noise vector, we have:
$$
\mathbb{E}[n(t-1-s)n^T(t-1-r)]=
\begin{cases}
\sigma^2I_N, \;\;\;\; s=r\\
0, \;\;\;\; s\neq r
\end{cases},
$$
which yields:
\begin{align}
U&=\sigma^2\sum_{s=0}^{t-1}\mathrm{tr}\left(\mathbb{E}\left[Q^T(s)Q(s)-H\right]\right)\nonumber\\
&=\sigma^2\sum_{s=0}^{t-1} \mathrm{tr}\left(\mathbb{E}\left[Q^T(s)I_NQ(s)-Q^T(s)H Q(s)\right]\right)\nonumber\\
&=\sigma^2\sum_{s=0}^{t-1} \mathrm{tr}\left(\mathbb{E}\left[Q^T(s)\Omega Q(s)\right]\right).\label{eq:QOQ}
\end{align}
Then, we express the trace using the vectorization operator ($\mathrm{tr}(A^TB)=\mathrm{vec}^T(A)\mathrm{vec}(B)$) as:
$$
U=\sigma^2\sum_{s=0}^{t-1} \vect^T(I_N)\vect\left(\mathbb{E}\left[Q^T(s)\Omega Q(s)\right]\right).
$$
Due to the idempotence of $\Omega$ and commutativity of the product $P(s)\Omega$, we can write the matrix $Q^T(s)\Omega Q(s)$ as:
\begin{equation}\label{eq:QOmegaQ}
Q^T\! (s)\Omega Q(s)\!=\! P(0)P(1)\ldots P(s\! -\! 1)P(s\! -\! 1)\Omega\ldots P(1)\Omega P(0)\Omega.      
\end{equation}
Then, we can apply the property $\vect(ABC)=(C^T\otimes A)\vect(B)$ and we get:
\begin{multline*}
U=\sigma^2\vect^T(I_N)\Bigg((I_N\otimes \Omega)\vect(I_N)\\  +\sum_{s=1}^{t-1}\!\mathbb{E}\left[(P(0)\Omega\otimes P(0)) \vect(P(1)\ldots P^2\! (s\! -\! 1)\Omega\ldots P(1)\Omega)\right]\!\Bigg)\! .
\end{multline*}
Following an iterative procedure we obtain:
\begin{multline*}
U=\sigma^2\vect^T(I_N)\Bigg((I_N\otimes \Omega)\\
+\sum_{s=1}^{t-1}\!\mathbb{E}\left[(P(0)\Omega\otimes P(0))\ldots(P(s\!-\! 1)\Omega\otimes P(s\! -\! 1))\right]\Bigg)\!\vect(I_N).
\end{multline*}
Since the matrices $P(s)$ are i.i.d., we get:
\begin{multline*}
U=\sigma^2\vect^T(I_N)\left(\sum_{s=1}^{t-1}\mathbb{E}[P\Omega\otimes P]^s+(I_N\otimes \Omega)\right)\vect(I_N) \\
=\sigma^2\vect^T(I_N)\left(\sum_{s=0}^{t-1}\mathbb{E}[P\Omega\otimes P]^s+(I_N\otimes \Omega)-I_{N^2}\right)\vect(I_N),
\end{multline*}
where to satisfy space constraints, we have avoided including the time dependence of matrices $P(t)$. This abuse of notation will be made multiple times in the rest of the paper.
Notice that $(I_N\otimes\Omega)-I_{N^2}=-(I_N\otimes\nones)$ and $\vect^T(I_N)(I_N\otimes\nones)\vect(I_N)=1$.
Let us recall $\K=\mathbb{E}[P\Omega\otimes P]$ and consider the limit $t\rightarrow{+\infty}$. By applying Jensen inequality 
to the spectral norm $\Vert \cdot \Vert_2$, we can see that $\Vert \mathbb{E}[P\Omega\otimes P]\Vert_2^2\le \mathbb{E}[\Vert P\Omega\otimes P \Vert_2^2]=\mathbb{E}[\Vert P\Omega \Vert_2^2\Vert P \Vert_2^2]<1$. 
Therefore,  $\sum_{s=0}^\infty \K^s=(I_{N^2}-\K)^{-1}$, and we have:
$$
\lim_{t\rightarrow{+\infty}}U=\sigma^2\left(\vect^T(I_N)(I_{N^2}-\K)^{-1}\vect(I_N)-1\right).
$$
The second term of \eqref{eq:3terms} is zero since the noise process has zero mean. The first term of \eqref{eq:3terms}, denoted by $F=\mathbb{E}\left[\Vert (Q(t)-H)x(0)\Vert^2\right]$,  can be expressed as:
\begin{align*}
F&=\mathbb{E}\left[\left[(Q(t)-H)x(0)\right]^T\left[(Q(t)-H)x(0)\right]\right]\\
&=\mathbb{E}\left[x^T(0)(Q^T(t)Q(t)-H)x(0)\right].
\end{align*}
Being $Q^T(t)Q(t)$ the product of doubly stochastic matrices, per Theorem~4 in \cite{tahbaz2009consensus} it converges to $H$ with probability 1, so that $\lim_{N\rightarrow{ \infty}} F=0$. 
Finally, $J_{\mathrm{noise}}$ is given by \eqref{eq:jnoise}.
\end{proof}

\begin{remark}[Alternate expressions for $J_{\mathrm{noise}}$]
Since $\Omega$ is idempotent and commutes with any stochastic matrix, expression~\eqref{eq:QOmegaQ} could be arranged in two other ways, so that instead of \eqref{eq:jnoise} one would obtain the analogous expression where the symmetric matrix $K=\mathbb{E}[P(t)\Omega\otimes P(t)]$  is replaced by $\mathbb{E}[P(t)\Omega\otimes P(t)\Omega]$ or by $\mathbb{E}[P(t)\otimes P(t)\Omega]$. 
\end{remark}

\subsection{Upper and lower bounds}
The calculation of \eqref{eq:jnoise} requires the determination of the $N^2\times N^2$ matrix $\K$
which is impractical either for theoretical or numerical analysis.
For this reason, we derive a lower bound $J_{\mathrm{LB}}$ and an upper bound $J_{\mathrm{UB}}$ for the noise index, which depend on the spectrum of symmetric $N$-dimensional matrices. 

\begin{theorem}[Upper and lower bounds]\label{thm:bounds}
For the noise index \eqref{eq:noise_index} we have:
\begin{equation}\label{eq:bounds}
\dfrac{\sigma^2}{N}\sum_{i=2}^N\dfrac{1}{1-\lambda_i^2(\bar P)} \le J_{\mathrm{noise}}\le \dfrac{\sigma^2}{N}\sum_{i=2}^N\dfrac{1}{1-\lambda_i(\bar{\bar P})} ,
\end{equation}
where $\bar P=\mathbb{E}[P(t)]$ and $\bar{\bar P}=\mathbb{E}[P^2(t)]$.
\end{theorem}

\begin{proof}
In order to obtain the lower bound for the noise index, we can express the 
term \eqref{eq:QOQ} as: 
$$
U=\sigma^2\sum_{s=0}^{t-1} \mathrm{tr}\left(\mathbb{E}\left[Q^T(s)\Omega Q(s)\right]\right)=
\sigma^2\sum_{s=0}^{t-1} \mathbb{E}\left[\Vert Q(s)\Omega \Vert_F^2\right],
$$
where $\Vert \cdot \Vert_F$ is the Frobenius norm. Then we can apply Jensen inequality 
obtaining:
$$
U\ge \sigma^2\sum_{s=0}^{t-1} \Vert \mathbb{E}\left[Q(s)\Omega\right]\Vert_F^2.
$$
Since the matrices $P(s)$  are i.i.d., we can see that $\mathbb{E}[Q(s)\Omega]=\bar P^s \Omega$ such that:
$$
U\ge \sigma^2\sum_{s=1}^{t-1} \Vert \bar P^s \Omega\Vert_F^2+\Vert \Omega \Vert_F^2.
$$
For a symmetric matrix $P$ we have $\Vert P\Vert_F^2=\Vert P^2 \Vert_{S_1}$, where $\Vert A \Vert_{S_1}=\sum_{i=1}^N \vert \lambda_i(A)\vert$ is the Schatten 1-norm. Thus:
$$
U\ge
\sigma^2\sum_{s=1}^{t-1} \Vert (\bar P^2 \Omega)^s\Vert_{S_1}+\Vert \Omega \Vert_{S_1}\ge\sigma^2\left\Vert \sum_{s=1}^{t-1}  (\bar P^2 \Omega)^s+ \Omega \right\Vert_{S_1}.
$$
When $t\rightarrow{+\infty}$, we have:
\begin{align*}
\lim_{t\rightarrow{+\infty}}U&\ge \sigma^2\left\Vert \sum_{s=0}^{\infty}  (\bar P^2 \Omega)^s+ \Omega -I_N\right\Vert_{S_1}\\
&=\sigma^2\left\Vert (I_N-\bar P^2\Omega)^{-1}-\nones\right\Vert_{S_1}.
\end{align*}
By using the definition of the Schatten 1-norm we get:
$$
\lim_{t\rightarrow{+\infty}}U\ge \sigma^2\sum_{i=2}^N\dfrac{1}{1-\lambda_i^2(\bar P)},
$$
which is the desired lower bound.

For the upper bound, we consider the summation element in \eqref{eq:QOQ}, denoted by $S$:
\begin{align*}
S=&\mathrm{tr}\left(\mathbb{E}\left[Q^T(s)(I_N-H)Q(s)\right]\right)\\
=&\mathrm{tr}\left(\mathbb{E}\left[Q^T(s)Q(s)\right]\right)-\mathrm{tr}\left(H\right).
\end{align*}
At this point, we can use the inequality 
$$\mathrm{tr}(\mathbb{E}[Q^T(s)Q(s)])\le \mathrm{tr}(\mathbb{E}[P^2]^s),$$ proved in 
\cite[Cor.~5]{del2011majorization}, to obtain:
\begin{align*}
S\le&\mathrm{tr}(\bar{\bar P}^s)-\mathrm{tr}\left(H\right)\\
=&\mathrm{tr}\left(\mathbb{E}\left[P^2(0)P^2(1)\cdots P^2(s-1)-H\right]\right)\\
=&\mathrm{tr}\left(\mathbb{E}\left[P^2(0)P^2(1)\cdots P^2(s-1)\Omega\right]\right).
\end{align*}
Since $\Omega$ is idempotent and $P(s)\Omega=\Omega P(s)$, we can express the latter as:
\begin{align*}
S\le&\mathrm{tr}\left(\mathbb{E}\left[P(0)\Omega P(0) 
\cdots P(s-1)\Omega P(s-1)\right]\right)\\
=&\mathrm{tr}\left(\mathbb{E}\left[P(0)\Omega P(0)\right] 
\cdots \mathbb{E}\left[P(s-1)\Omega P(s-1)\right]\right)\\
=&\mathrm{tr}\left(\mathbb{E}\left[P(t)\Omega P(t)\right]^s\right).
\end{align*}
Then, the sum in \eqref{eq:QOQ} is upper bounded by:
$$
U\le \sigma^2\sum_{s=1}^{t-1} \mathrm{tr}\left(\mathbb{E}\left[P(0)\Omega P(0)\right]^s\right).
$$
When $t\rightarrow{+\infty}$, we get:
\begin{align*}
\lim_{t\rightarrow{+\infty}}U\le& \sigma^2\mathrm{tr}\left(\sum_{s=0}^{\infty} \mathbb{E}\left[P\Omega P\right]^s-H\right)\\
=&\sigma^2\mathrm{tr}\left((I_N-\mathbb{E}\left[P\Omega P\right])^{-1}-H\right)\\
=&\sigma^2\sum_{i=2}^N\dfrac{1}{1-\lambda_i(\bar{\bar P})},
\end{align*}
which gives the desired upper bound.
\end{proof}
\begin{remark}[Deterministic case]
For the deterministic dynamics $x(t+1)={\bar P}x(t)+n(t)$, that is, $P(t)=\bar P$, 
the noise index is given by $J_{\mathrm{noise}}=\frac{\sigma^2}{N}\sum_{i=2}^N\frac{1}{1-\lambda_i^2(\bar P)}$. In this case the upper and lower bounds of Theorem~\ref{thm:bounds} coincide and recover the well-known results about deterministic consensus  \cite{garin2010survey,fagnani2017introduction}. These questions have also been extensively studied in  continuous-time by the closely-related notion of network coherence  \cite{patterson2010leader,bamieh2012coherence}.
\end{remark}

\section{Randomly Induced Discretized Laplacians}

We consider the following sampling procedure. Let a set~$V$ of $N$ nodes be given, together with an undirected graph $\bar G$ that has $V$ as node set and whose adjacency matrix is denoted by $\bar A$. At each time $t$, every node has an independent probability $p$ of being {\em active}; at each time $t$, the current graph $G(t)$ is the subgraph of $\bar G$ that is induced by the set of active nodes. More formally, the current adjacency matrix $A(t)$ is generated as
\begin{equation}\label{eq:varying_A}
A(t)=\Gamma(t)\, \bar A\, \Gamma(t),  
\end{equation}
where   $\Gamma(t)=\mathrm{diag}[\gamma_1(t),\ldots,\gamma_N(t)]$ is a diagonal matrix with $N$ independent identically distributed Bernoulli random variables corresponding to each node. The value of the random variable $\gamma_i(t)$ indicates the state of agent $i$ at time $t$ such that when $\gamma_i(t)=1$, the node is active. 

Given the adjacency matrix $A(t)$, we then define the time-varying doubly stochastic matrix given by:
\begin{equation}\label{eq:varying_stochastic}
P(t)=I_N-\epsilon (D(t)- A(t)),    
\end{equation}
where $D(t)$ is the diagonal matrix of the degrees of the induced graph $G(t)$ and $0<\epsilon<\frac{1}{d_{\max}(\bar G)}$. 
After defining the time-varying Laplacian $L(t)=D(t)-A(t)$, it becomes natural to refer to stochastic matrices sampled according to  \eqref{eq:varying_stochastic}-\eqref{eq:varying_A} as {\em Randomly Induced Discretized Laplacians (RIDL)}.

\begin{remark}[RIDL and Open Multi-Agent Systems]
When studying Open Multi-Agent Systems, one is confronted with the issue of agents continuously joining and leaving the system. When the pool of possible participating agents is known,  arrivals and departures can be equivalently seen as activations and de-activations~\cite{varma2018open}. In such a case, one can further assume that the interaction network is a subgraph, induced by the active nodes, of a larger network of potential interactions. 
RIDL can therefore be used to study OMAS in this case, under the assumption that the activation of each agent is an independent event, as illustrated in Fig.~\ref{fig_random}.      
\end{remark}

\begin{figure}
\centering
\includegraphics[width=\columnwidth]{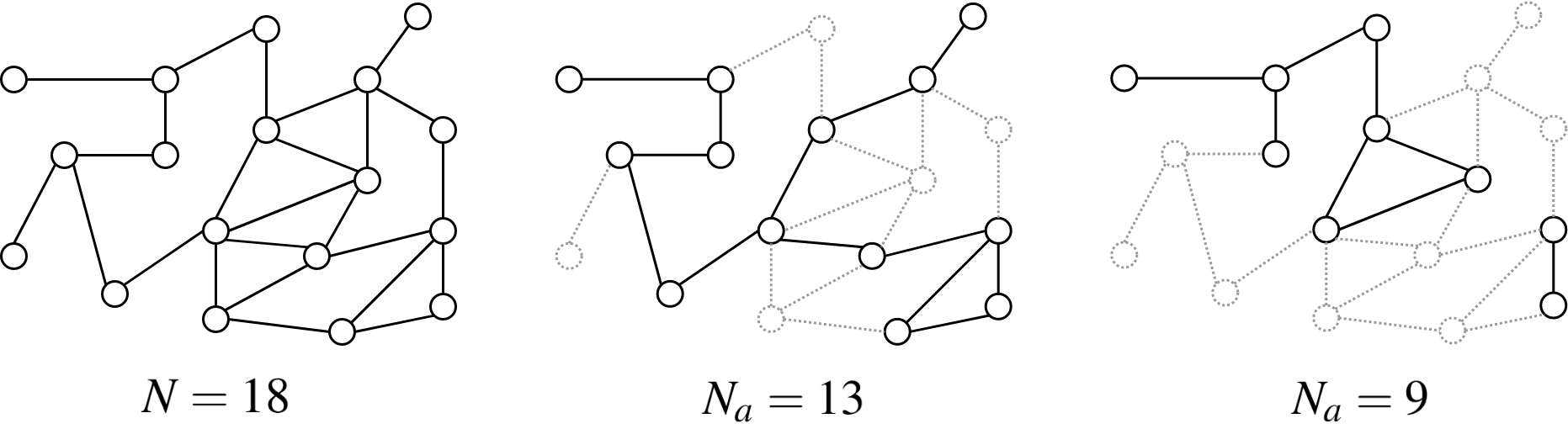} 
\caption{Random interactions with different numbers of active agents $N_a$ in a network with $N=18$ possible agents.}
\label{fig_random}
\end{figure}

We now study $J_\mathrm{noise}$ for RIDL matrices.
For an exact calculation as per \eqref{eq:jnoise}, it would be necessary to compute the matrix $\K=\mathbb{E}[P(t)\Omega\otimes P(t)]$. To this end, we can observe that  
$$
\K=
\mathbb{E}[P(t)\otimes P(t)]-\nones\otimes \bar P.
$$
The second term only depends on $\bar P$, which is given by
\begin{equation}\label{eq:Pexpected}
\bar P=I_N-\epsilon p^2(\bar D-\bar A)=I_N-\epsilon p^2\bar L,
\end{equation}
where $\bar D$ and $\bar L$ are the degree and Laplacian matrices of graph $\bar G$, but the first term
$$
\mathbb{E}[P(t)\otimes P(t)]=I_{N^2}-\epsilon p^2(I_N\otimes \bar L+\bar L\otimes I_N)+\epsilon^2 \mathbb{E}[L(t)\otimes L(t)]
$$
cannot be written as a function of $\bar P$ or $\bar L$. Hence, $J_\mathrm{noise}$ cannot be written as a function of $\bar P$. Nevertheless, the upper and lower bounds of Theorem 2 can be expressed as functions of the eigenvalues of $\bar P$.

\begin{prop}[Bounds for Discretized Laplacian]\label{prop:bounds}
For the time-varying stochastic matrix \eqref{eq:varying_stochastic} generated with \eqref{eq:varying_A}, the lower bound $J_\mathrm{LB}$ and upper bound $J_\mathrm{UB}$ in  \eqref{eq:bounds} become:
\begin{gather}
J_\mathrm{LB}=\dfrac{\sigma^2}{\epsilon p^2N}\sum_{i=2}^{N}\dfrac{1}{2\lambda_i(\bar L)-\epsilon p^2 \lambda_i^2(\bar L)}, \label{eq:lower} \\   
J_\mathrm{UB}=\dfrac{\sigma^2}{\epsilon p^2 N}\sum_{i=2}^{N}\dfrac{1}{2(1+\epsilon p-\epsilon)\lambda_i(\bar L)-\epsilon p \lambda_i^2(\bar L)}, \label{eq:upper}
\end{gather}
where $0=\lambda_1(\bar L)<\lambda_2(\bar L)\le \ldots \le \lambda_N(\bar L)$ are the eigenvalues of $\bar L$.
\end{prop}

\begin{proof}
For the lower bound, we consider the expression \eqref{eq:Pexpected}. Since $\bar L$ is symmetric, the eigenvalues of $\bar P$ are given by $\lambda_i(\bar P)=1-\epsilon p^2\lambda_i(\bar L)$ and we obtain \eqref{eq:lower}.

For the upper bound, $\bar{\bar P}$ can be written as:
\begin{align*}
\bar{\bar P}&=\mathbb{E}[(I_N-\epsilon L)^2]
=I_N-2\epsilon p^2 \bar L+\epsilon^2\mathbb{E}[L^2]\\
&=I_N-2\epsilon p^2 \bar L+\epsilon^2(\mathbb{E}[D^2]-\mathbb{E}[D A]-\mathbb{E}[A D]+\mathbb{E}[A^2])
\end{align*}
Since $\mathbb{E}[\gamma_i(t)]=\mathbb{E}[\gamma_i^2(t)]=p$, we have:
$$
\mathbb{E}[D^2]=(p^2-p^3)\bar D+p^3 \bar D^2, \quad \mathbb{E}[DA]=p^3\bar D \bar A+(p^2-p^3)\bar A,
$$
$$
\mathbb{E}[AD]=p^3\bar A\bar D+(p^2-p^3)\bar A, \quad \mathbb{E}[A^2]=(p^2-p^3)\bar D+p^3 \bar A^2,
$$
which yields:
$$
\bar{\bar P}=I_N+2\epsilon p^2(\epsilon-\epsilon p-1)\bar L+\epsilon^2 p^3 \bar L^2.    
$$
The eigenvalues of $\bar{\bar P}$ are
$\lambda_i(\bar{\bar P})=1+2\epsilon p^2(\epsilon-\epsilon p-1)\lambda_i(\bar L)+\epsilon^2 p^3 \lambda_i^2(\bar L)$ and we obtain \eqref{eq:upper}.
\end{proof}

From expressions \eqref{eq:lower} and \eqref{eq:upper} we can observe that for the computation of the bounds for $J_{\mathrm{noise}}$, we only need to know the topology of 
$G_{\bar P}$
through $\bar L$.

Furthermore, inspired by the work in \cite{lovisari2013resistance}, we relate $J_{\mathrm{noise}}$ with
the well-studied \textit{Average Effective Resistance} or {\em Kirchhoff Index} of the graph. The latter can indeed be expressed \cite[Chapter~5]{fagnani2017introduction} as a function of Laplacian eigenvalues by
$$
R_{\mathrm{ave}}:=\dfrac{1}{N}\sum_{i=2}^N \dfrac{1}{\lambda_i(\bar L)}.
$$
\begin{prop}[Bounds and average effective resistance]\label{prop:average}
For the noise index \eqref{eq:noise_index} with the time-varying stochastic matrix \eqref{eq:varying_stochastic} generated with \eqref{eq:varying_A} we have: 
$$
\dfrac{\sigma^2}{2 p^2} \dfrac{R_{\mathrm{ave}}}{\epsilon}\le J_{\mathrm{noise}}\le \dfrac{\sigma^2}{2p^3(1-k )}\dfrac{R_{\mathrm{ave}}}{\epsilon},
$$
where $k=\epsilon d_{\max}<1$.
\end{prop}
\begin{proof}
For the lower bound, we immediately observe from \eqref{eq:lower} that
$$
 J_{\mathrm{LB}}\ge \dfrac{\sigma^2}{\epsilon p^2 N}\sum_{i=2}^N\dfrac{1}{2\lambda_i(\bar L)}=\dfrac{\sigma^2}{2\epsilon p^2} R_{\mathrm{ave}}.
 $$
For the upper bound, we remind 
$\lambda_i(\bar L)\le 2 d_{\max}$ and obtain:
\begin{align*}
  \MoveEqLeft[3]  2(1+\epsilon p-\epsilon)\lambda_i(\bar L)-\epsilon p \lambda_i^2(\bar L)=\\ 
    &=2\epsilon(p-1)\lambda_i(\bar L)+(2-\epsilon p\lambda_i(\bar L))\lambda_i(\bar L)\\
    &\ge 2\epsilon(p-1)\lambda_i(\bar L)+\left(2-\epsilon p(2d_{\max})\right)\lambda_i(\bar L)\\
    &\ge 2p(1-k)\lambda_i(\bar L),
\end{align*}
which yields:
$$
J_{\mathrm{UB}}\le \dfrac{\sigma^2}{\epsilon p^2 N}\sum_{i=2}^N\dfrac{1}{2p(1-k)\lambda_i(\bar L)}=\dfrac{\sigma^2}{2\epsilon p^3 (1-k)}R_{\mathrm{ave}},
$$
thereby completing the proof.
\end{proof}
The bounds of Proposition~\ref{prop:average}
can also be equivalently expressed as
$$
\dfrac{\sigma^2}{2 p^2 k} d_{\max} R_{\mathrm{ave}}\le J_{\mathrm{noise}}\le \dfrac{\sigma^2}{2p^3k(1-k )}d_{\max} R_{\mathrm{ave}}.
$$
Since only $d_{\max}$ and $R_{\mathrm{ave}}$ depend on the graph (and therefore on $N$), this form of the bounds implies that   $J_{\mathrm{noise}}=\Theta(d_{\max}R_{\mathrm{ave}})$ as $N\to\infty$ and proves that the bounds are asymptotically tight for RIDLs.
For instance, for 
sparse graphs with bounded degree,
the rate of growth of $J_{\mathrm{noise}}$ is equal to the rate of growth of $R_{\mathrm{ave}}$. Instead, for dense graphs where $d_{\max}$ is linear in $N$, the rate of growth of $J_{\mathrm{noise}}$ is the rate of $NR_{\mathrm{ave}}$.

\section{Examples: sparse and dense graphs}
In this section, we illustrate the behavior of the noise index for RIDL matrices associated to graphs with particular structures, which allow for the explicit computation of the bounds in Theorem~\ref{thm:bounds} and Proposition~\ref{prop:bounds}.
We consider $\epsilon=k/d_{\max} (\bar G)$ with $k\in (0,1)$ and we perform computations for networks with $3\le N \le 100$, $p=0.9$, and $\sigma^2=1$.

\paragraph{Star graphs}
For the star graph $S_N$ we have $d_{\max}(S_N)=N-1$ and we choose $\epsilon=k/(N-1)$. The eigenvalues of the Laplacian matrix are given by $\lambda_i(\bar L)=1$ for $i=2,\ldots ,N-1$ and $\lambda_N(\bar L)=N$, so that the bounds are:
\begin{gather*}
J_\mathrm{LB} = \dfrac{\sigma^2}{\epsilon p^2N}\left(\dfrac{N-2}{2-\epsilon p^2}+\dfrac{1}{N(2-\epsilon p^2 N)}\right),\\
J_\mathrm{UB} =\dfrac{\sigma^2}{\epsilon p^2 N}\left(\dfrac{N-2}{\epsilon p+2-2\epsilon}+\dfrac{1}{N(2\epsilon p+2-2\epsilon-\epsilon p N)} \right),
\end{gather*}
implying that $J_{\mathrm{noise}}= \frac{\sigma^2}{2kp^2}N+o(N)
$ as $N\to\infty$. 
Figures~\ref{fig:star1} and \ref{fig:star2} present the results of the computations for the star graph with $k=0.8$.

\begin{figure}
\centering
\subfigure[$J_{\mathrm{noise}}$]
{ \includegraphics[width=4cm]{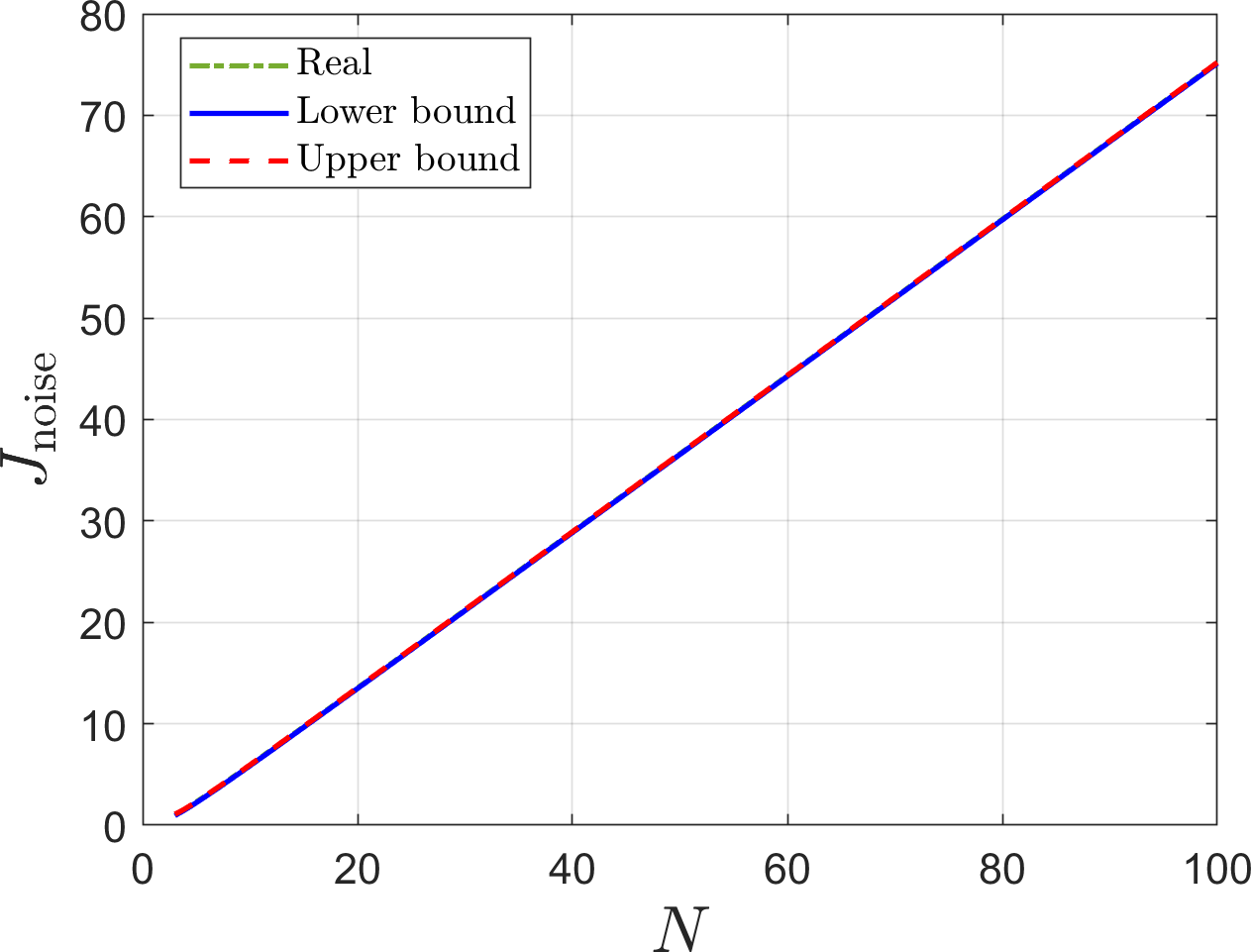}\label{fig:star1}
}
\subfigure[Relative error]
{ \includegraphics[width=4cm]{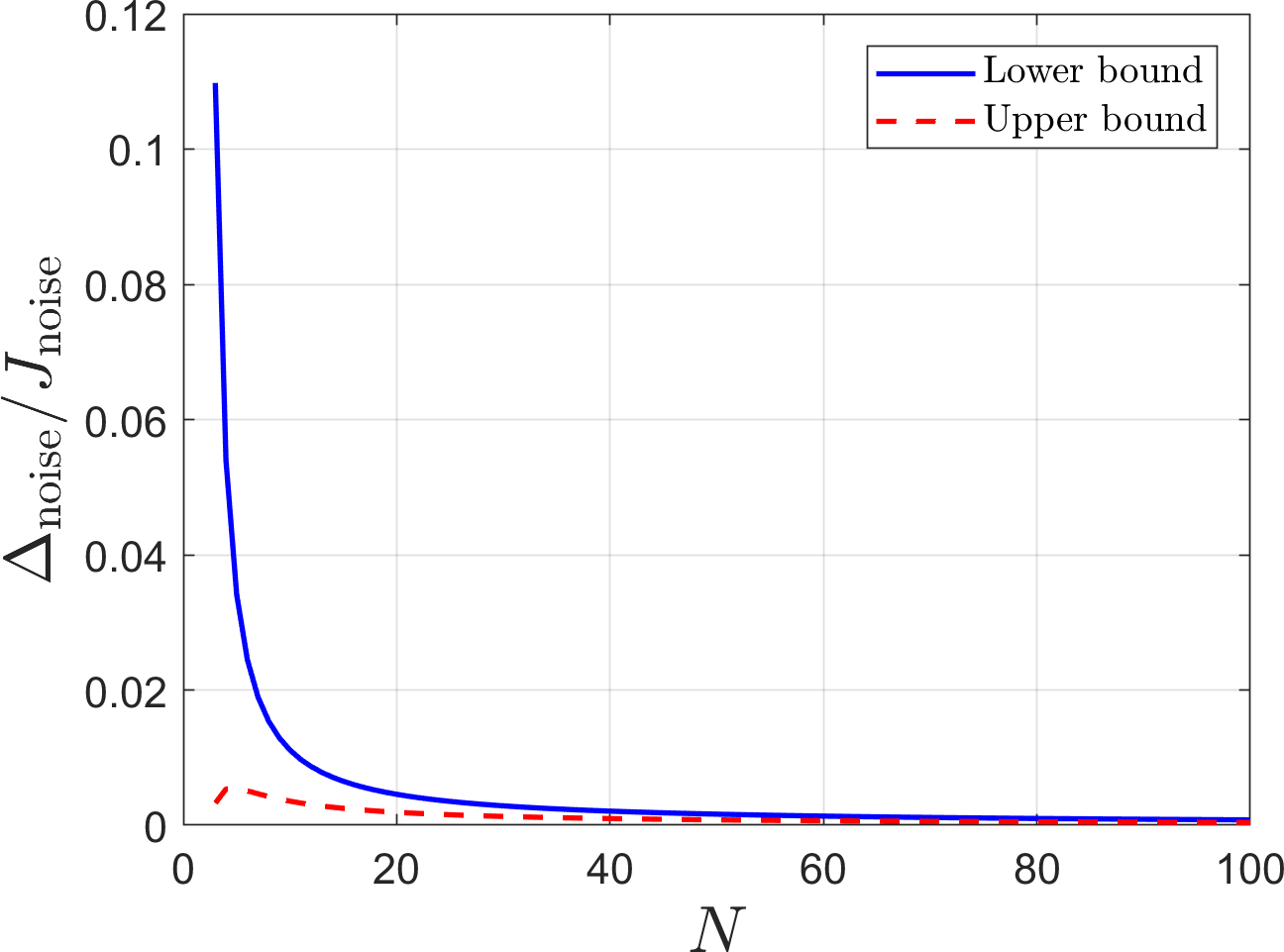}\label{fig:star2}
}
\caption{Computations for a star graph $S_N$ with growing $N$.}
\end{figure}

\paragraph{Paths and grids}
For the path graph $P_N$, $d_{\max}(P_N)=2$ and we select $\epsilon=0.8/2$ for all $N$. The eigenvalues of the Laplacian matrix are given by $\lambda_i(\bar L)=2-2\cos\left(\frac{\pi(i-1)}{N}\right)$ for $i=2,\ldots ,N$ such that the corresponding bounds are:
\begin{gather*}
J_\mathrm{LB}\!\!=\!\!\dfrac{\sigma^2}{4\epsilon p^2N}\!\sum_{i=1}^{N-1}\!\dfrac{1}{1\!-\!\epsilon p^2\!-\!\epsilon p^2\!\cos^2\!\left(\frac{\pi i}{N}\right)\!+\!(2\epsilon p^2\!-\!1)\!\cos\left(\frac{\pi i}{N}\right)},\\    
J_\mathrm{UB}\!\!=\!\!\dfrac{\sigma^2}{4\epsilon p^2N}\!\sum_{i=1}^{N-1}\!\dfrac{1}{1\!-\!\epsilon\!-\!\epsilon p\cos^2\!\left(\frac{\pi i}{N}\right)\!+\!(\epsilon p\!+\!\epsilon\!-\!1)\cos\left(\frac{\pi i}{N}\right)}
\end{gather*}
and the rate of growth of the noise index is $J_{\mathrm{noise}}= \Theta(N)$.
Fig.~\ref{fig:path1} and \ref{fig:path2} present the results of the computations for the path graph.
\begin{figure}
\centering
\subfigure[$J_{\mathrm{noise}}$]
{ \includegraphics[width=4cm]{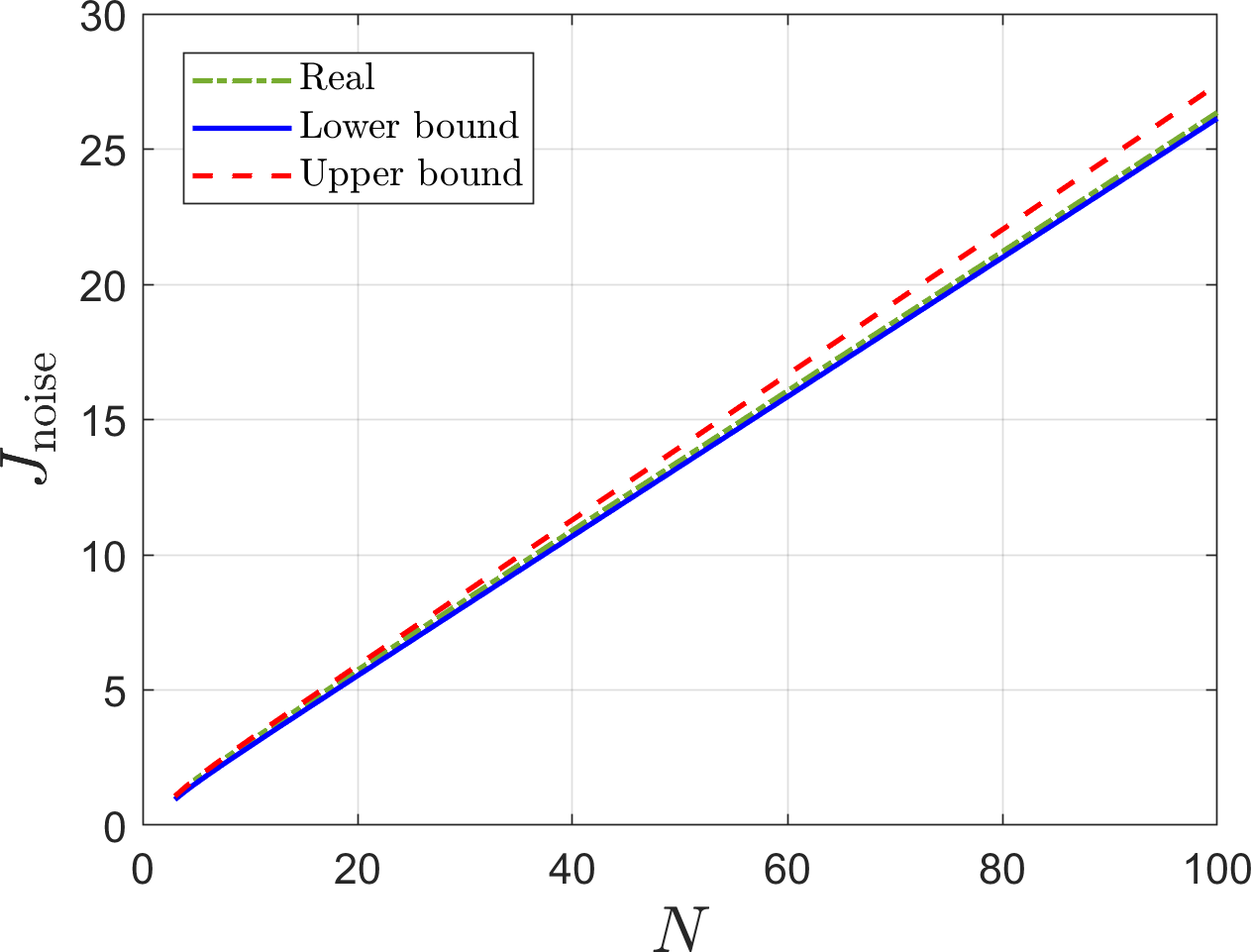}\label{fig:path1}
}
\subfigure[Relative error]
{ \includegraphics[width=4cm]{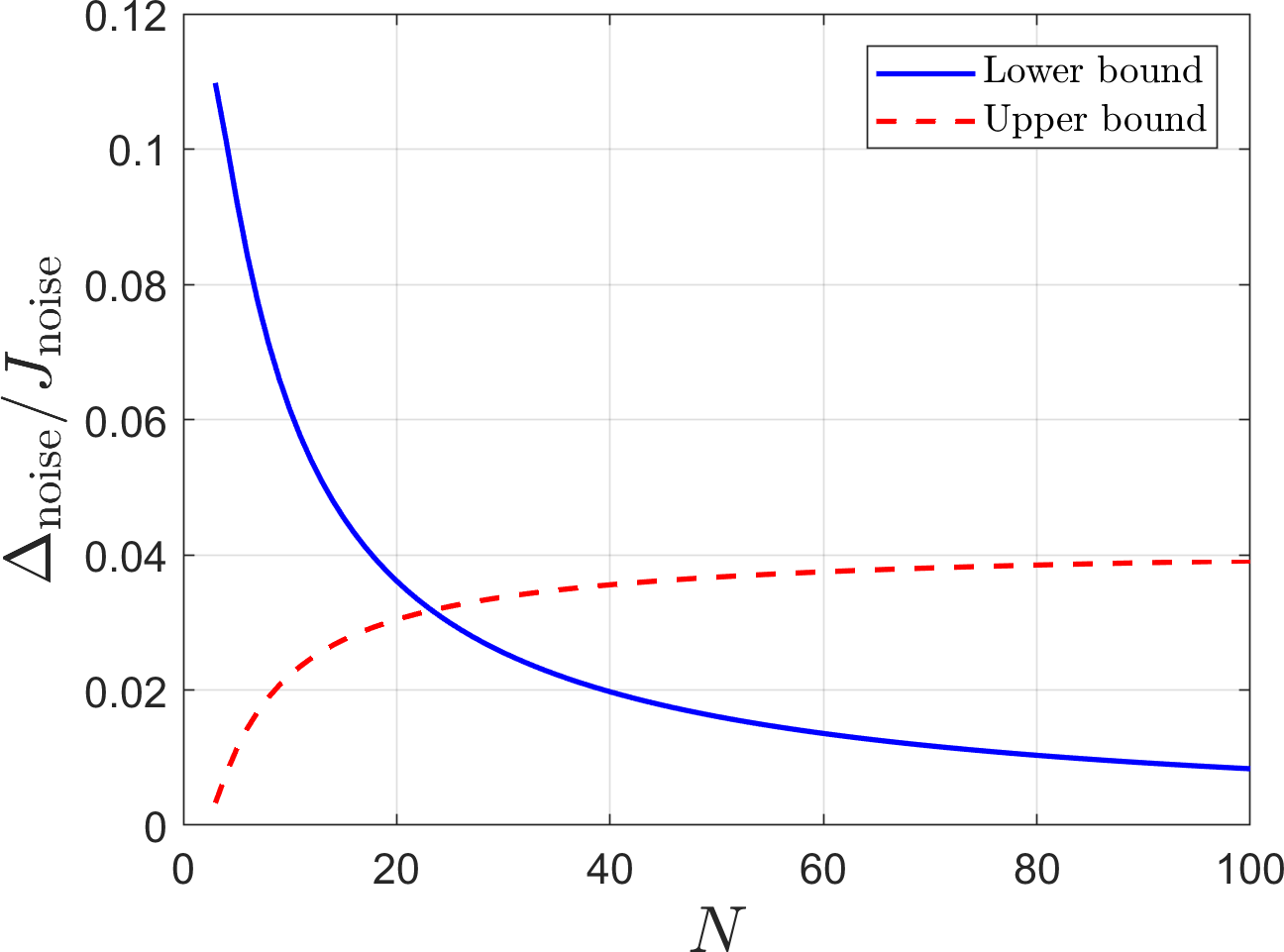}\label{fig:path2}
}
\caption{Computations for a path graph $P_N$ with growing $N$.}
\end{figure}

The maximum degree for a 2-D grid is $d_{\max}(\text{2-D})=4$ and for a 3-D grid is $d_{\max}(\text{3-D})=6$.  
The spectrum of the Laplacian matrix for a k-D grid are given by: $2k-2\sum_{i=1}^k\cos(\frac{\pi}{N}h_i)$ with $h_i\in\{0,\ldots, N-1 \}$.
Due to the complexity of dealing with the closed form from Theorem~1, we present only the behavior of the bounds for the 2-D grid in Fig~\ref{fig:2_d} and for the 3-D grid in Fig~\ref{fig:3_d}. These examples also confirm the theoretical analysis.

\begin{figure}
\centering
\subfigure[2-D grid]
{ \includegraphics[width=4cm]{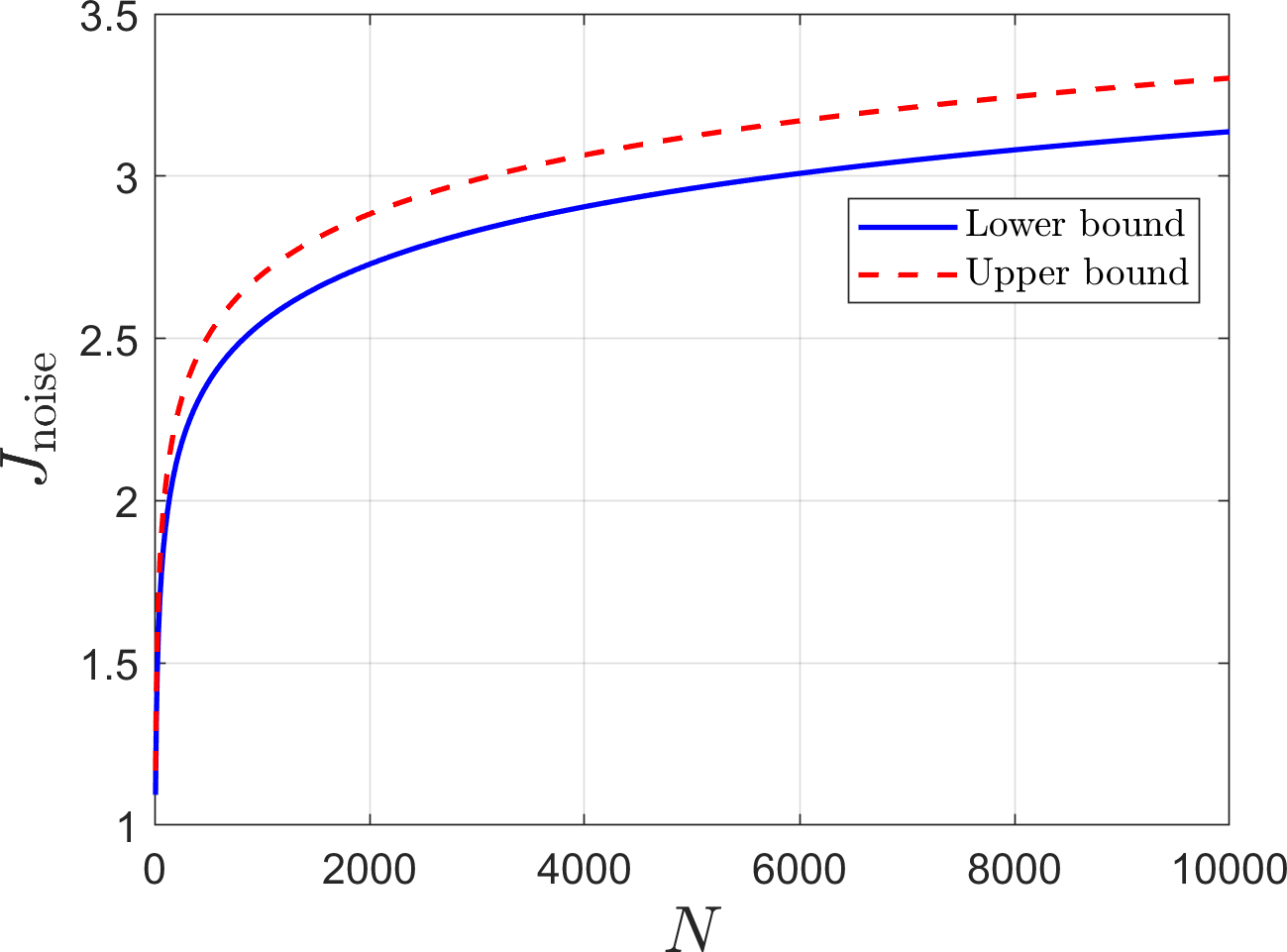}\label{fig:2_d}
}
\subfigure[3-D grid]
{ \includegraphics[width=4cm]{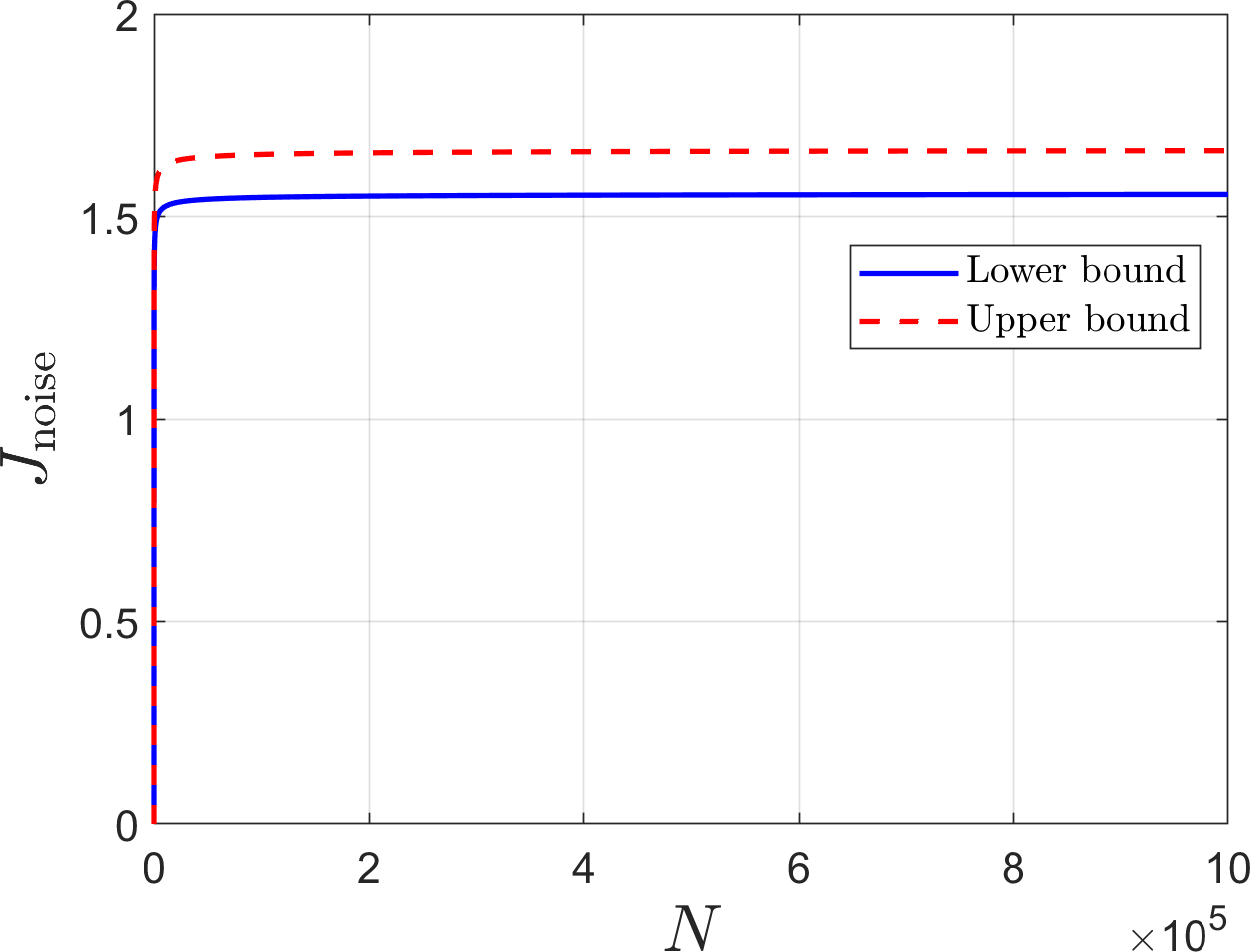}\label{fig:3_d}
}
\caption{Computations for 2$-$D and 3$-$D grids with growing $N$.}
\end{figure}

\paragraph{Dense graphs (complete and Erd\H{o}s-R\'{e}nyi)}
Since for a complete graph $K_N$, $d_{\max}(K_N)=N-1$, we select $\epsilon=k/(N-1)$.
The eigenvalues of the Laplacian matrix of the complete graph $K_N$ are $\lambda_i(\bar L)=N$ for $i=2,\ldots ,N$ such that the bounds for $J_{\mathrm{noise}}$ are:
\begin{gather*}
J_\mathrm{LB}=\dfrac{\sigma^2(N-1)}{\epsilon p^2N^2(2-\epsilon p^2 N)} \to \dfrac{\sigma^2}{p^2k(2-p^2k)} \text{as $N\to\infty$},\\
J_\mathrm{UB}=\dfrac{\sigma^2(N-1)}{\epsilon p^2 N^2(2\!+\!2\epsilon p\!-\!2\epsilon\!-\!\epsilon p N)}\to \dfrac{\sigma^2}{p^2k(2\!-\!pk)} \text{as $\!N\to\infty$},
\end{gather*}
implying a rate of growth $J_{\mathrm{noise}}=\Theta (1).$
For the computations we choose $k=0.8$ and we obtain the results in Figures \ref{fig:complete1} and \ref{fig:complete2}: the latter implies that the upper bound is equal to the noise index. 
\begin{figure}
\centering
\subfigure[$J_{\mathrm{noise}}$]
{ \includegraphics[width=4cm]{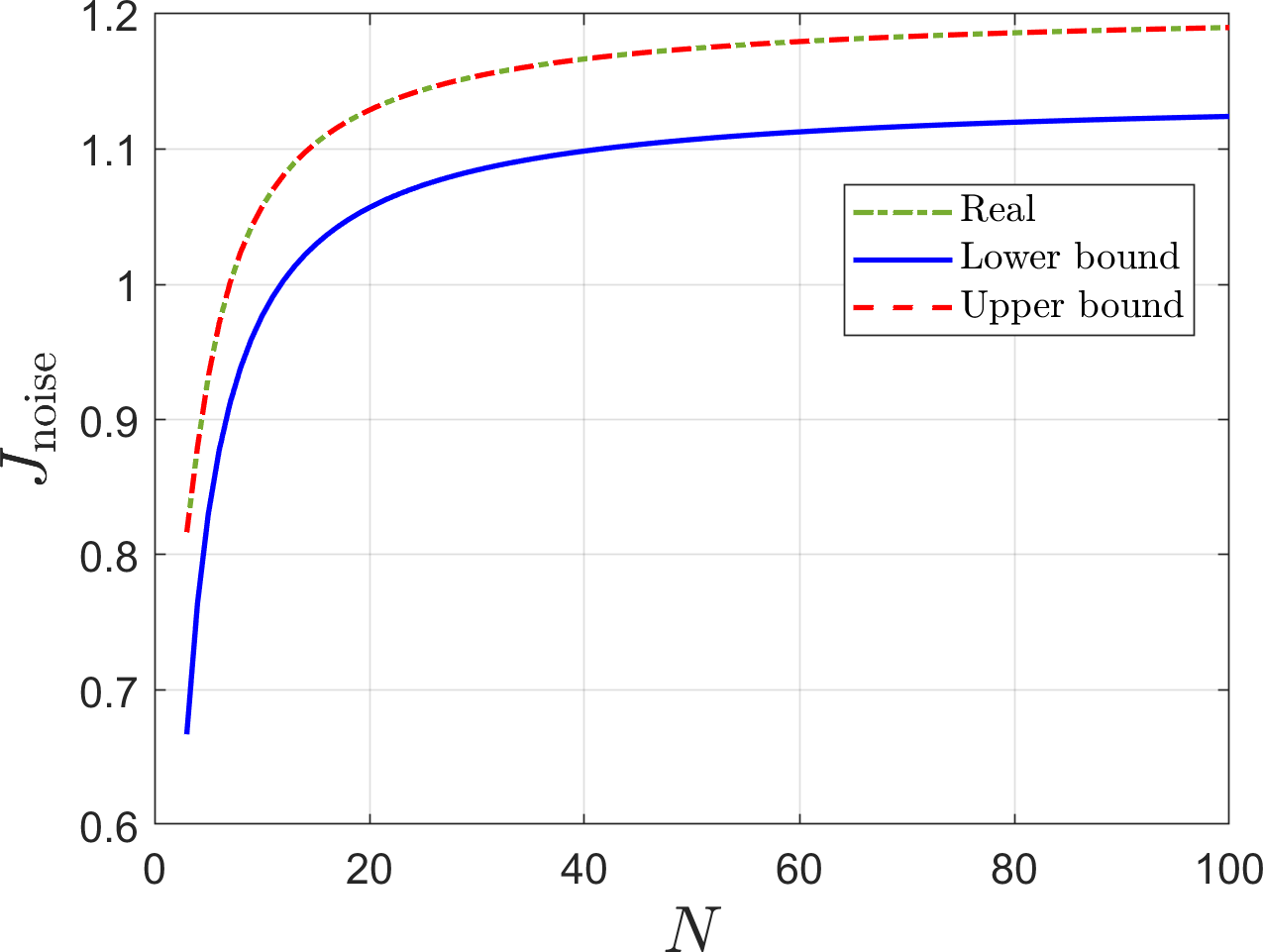}\label{fig:complete1}
}
\subfigure[Relative error]
{ \includegraphics[width=4cm]{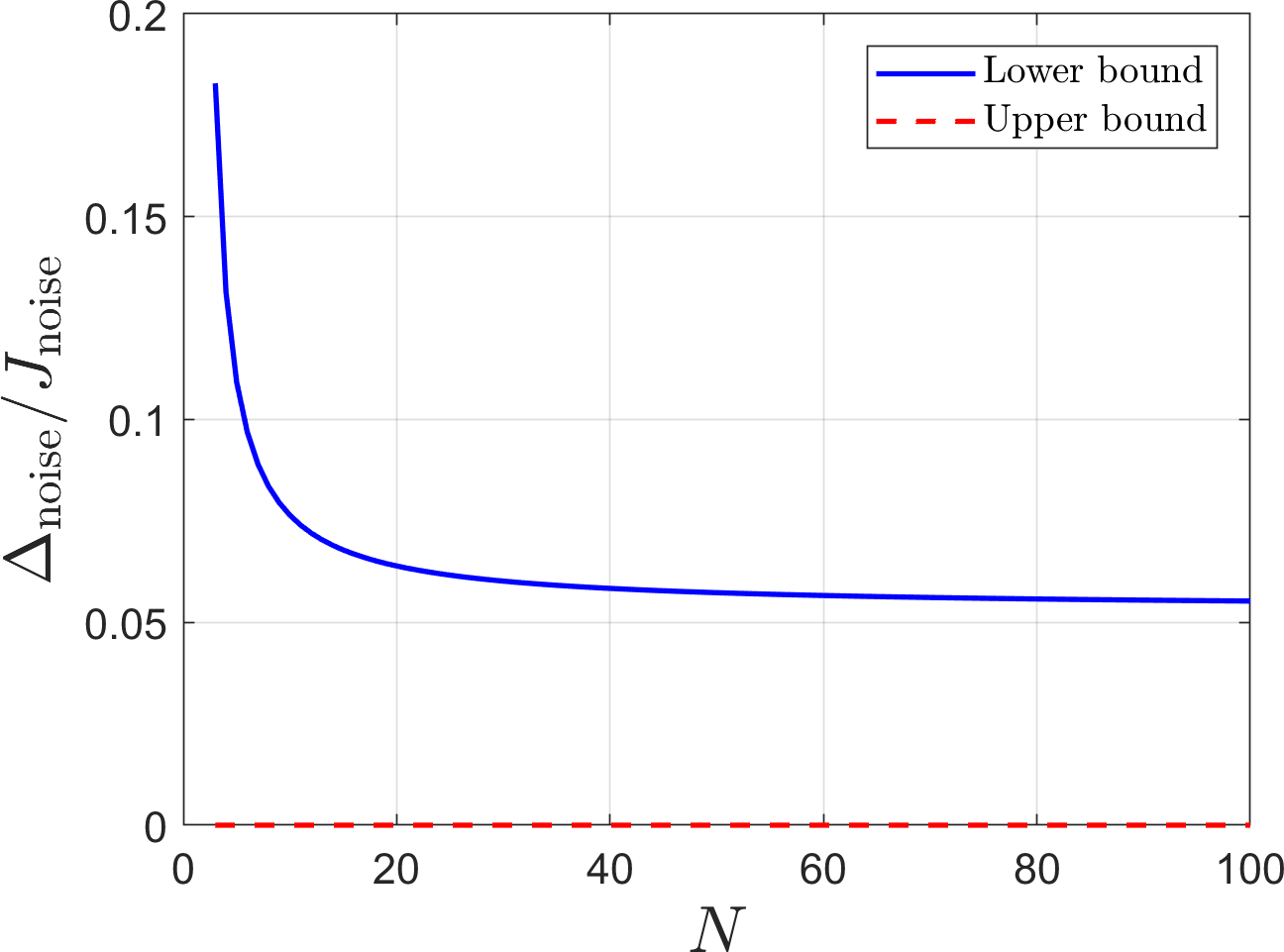}\label{fig:complete2}
}
\caption{Computations for a complete graph $K_N$ with growing $N$.}
\end{figure}
Simulations suggest that other dense graphs behave similarly to the complete. As an example, Figures \ref{fig:erdos1} and \ref{fig:erdos2} present the results of the computations for the Erd\H{o}s-R\'{e}nyi graph $p_\mathrm{er}=0.8$ and $\epsilon=0.8/(N-1)$.
\begin{figure}
\centering
\subfigure[$J_{\mathrm{noise}}$]
{ \includegraphics[width=4cm]{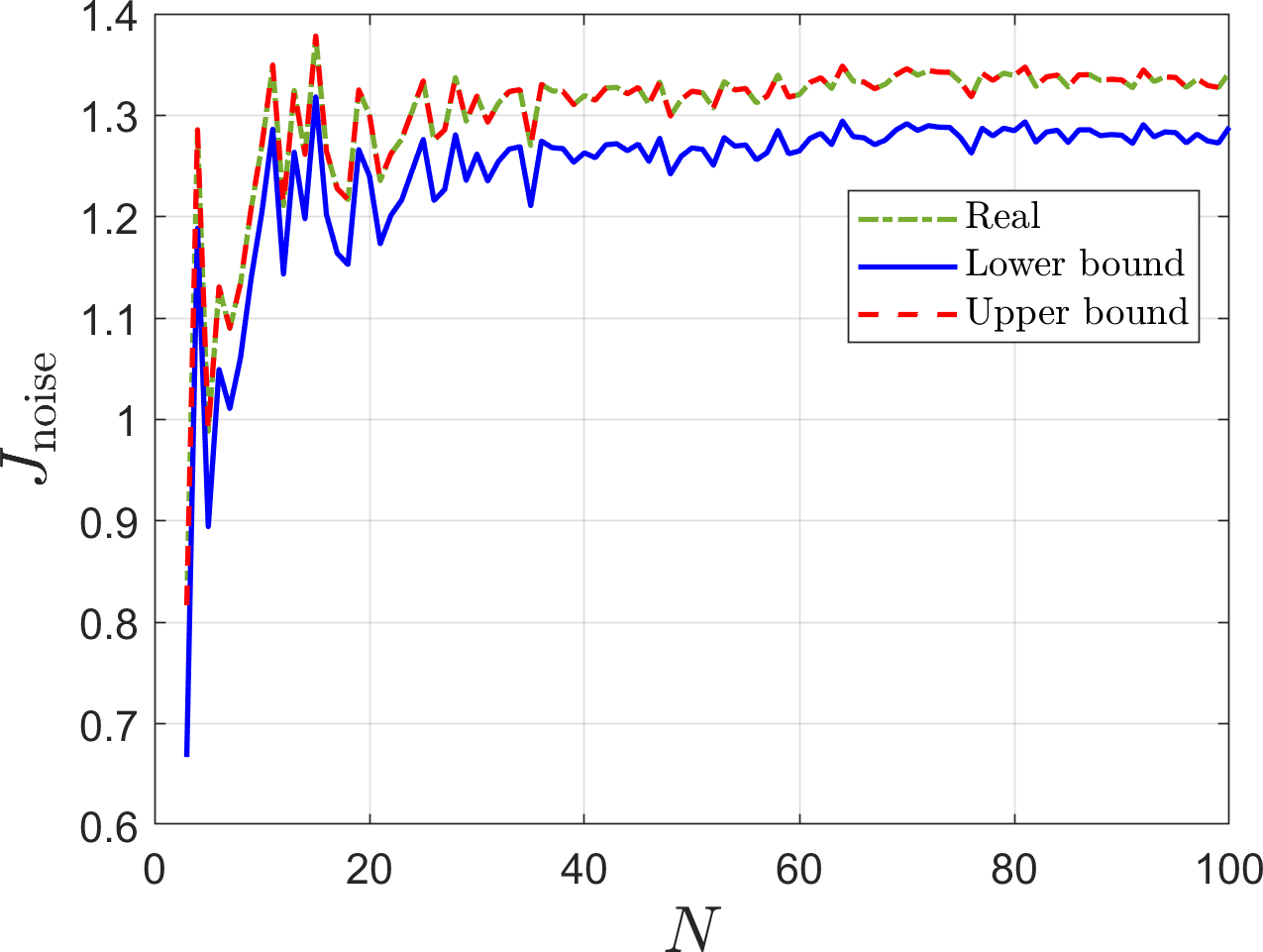}\label{fig:erdos1}
}
\subfigure[Relative error]
{ \includegraphics[width=4cm]{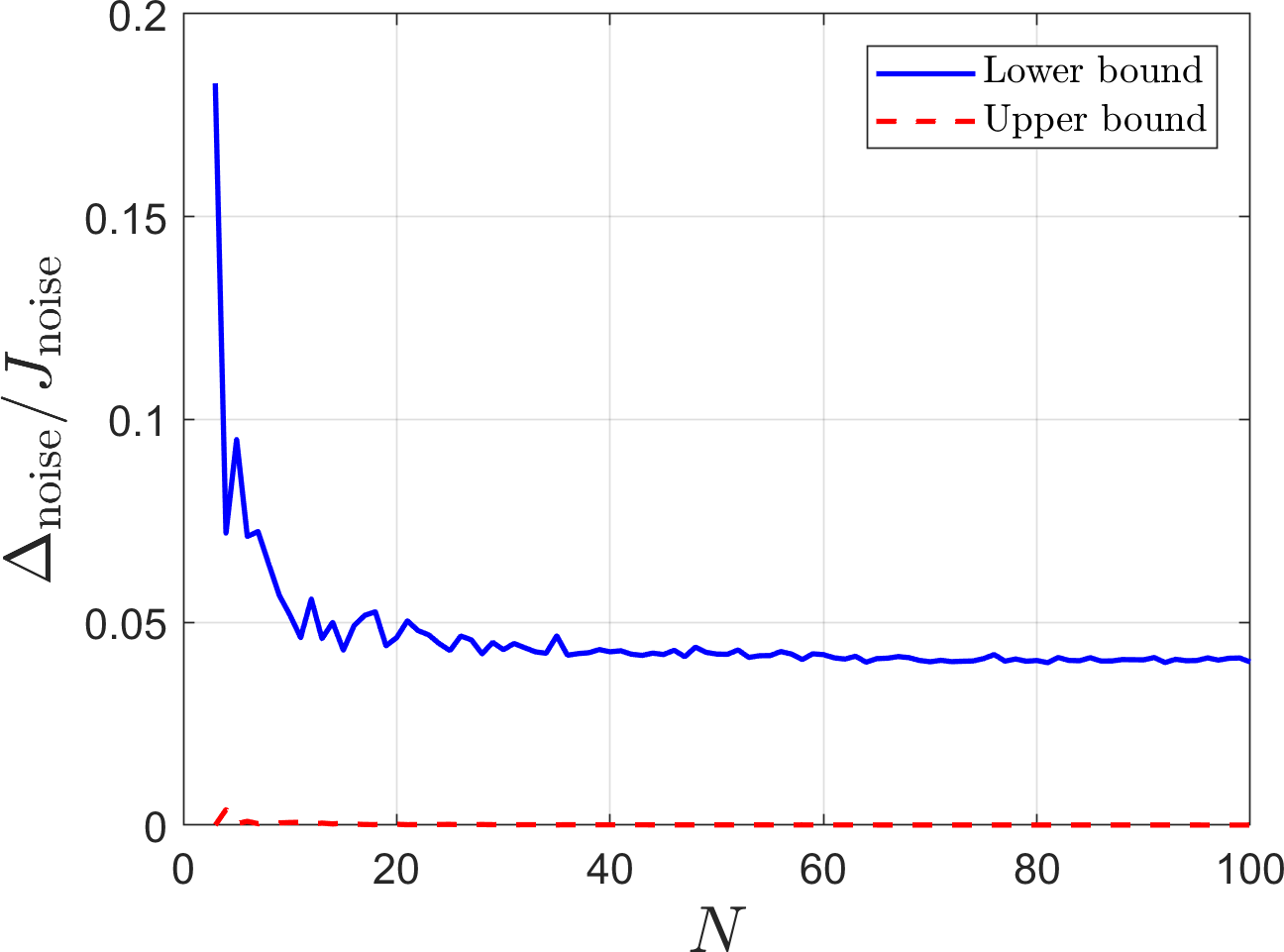}\label{fig:erdos2}
}
\caption{Computations for a Erd\H{o}s-R\'{e}nyi graph with growing $N$.}
\end{figure}

\paragraph{Variation of $p$}
Figures~\ref{fig:lower_bound_p} and \ref{fig:upper_bound_p} present the behavior of the relative error for the bounds of Theorem~\ref{thm:bounds} for graphs with $N=100$, $\epsilon=0.8/d_{\max}$ and $0.1\le p\le 0.9$.
\begin{figure}
\centering
\subfigure[Lower bound]
{ \includegraphics[width=4cm]{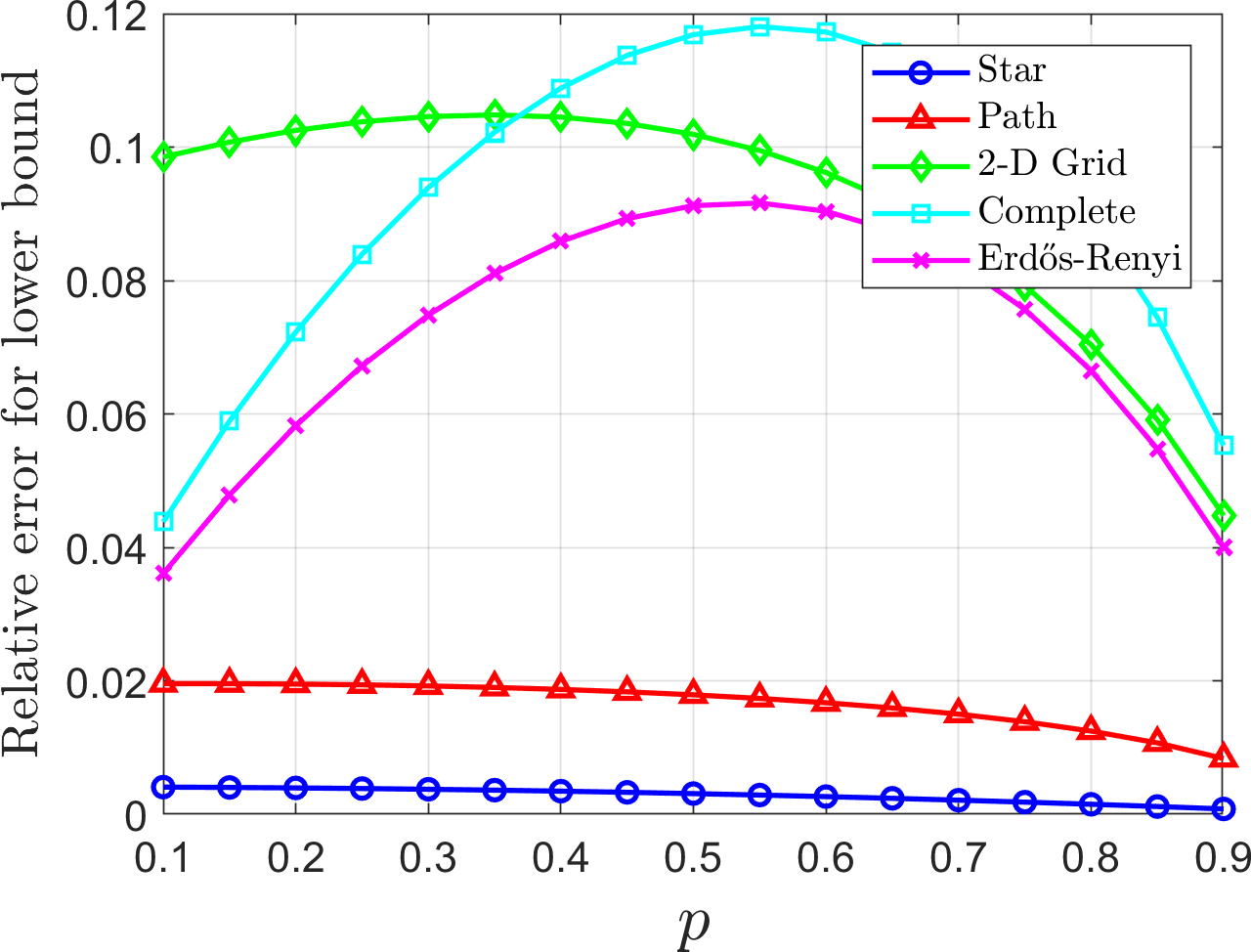}\label{fig:lower_bound_p}
}
\subfigure[Upper bound]
{ \includegraphics[width=4cm]{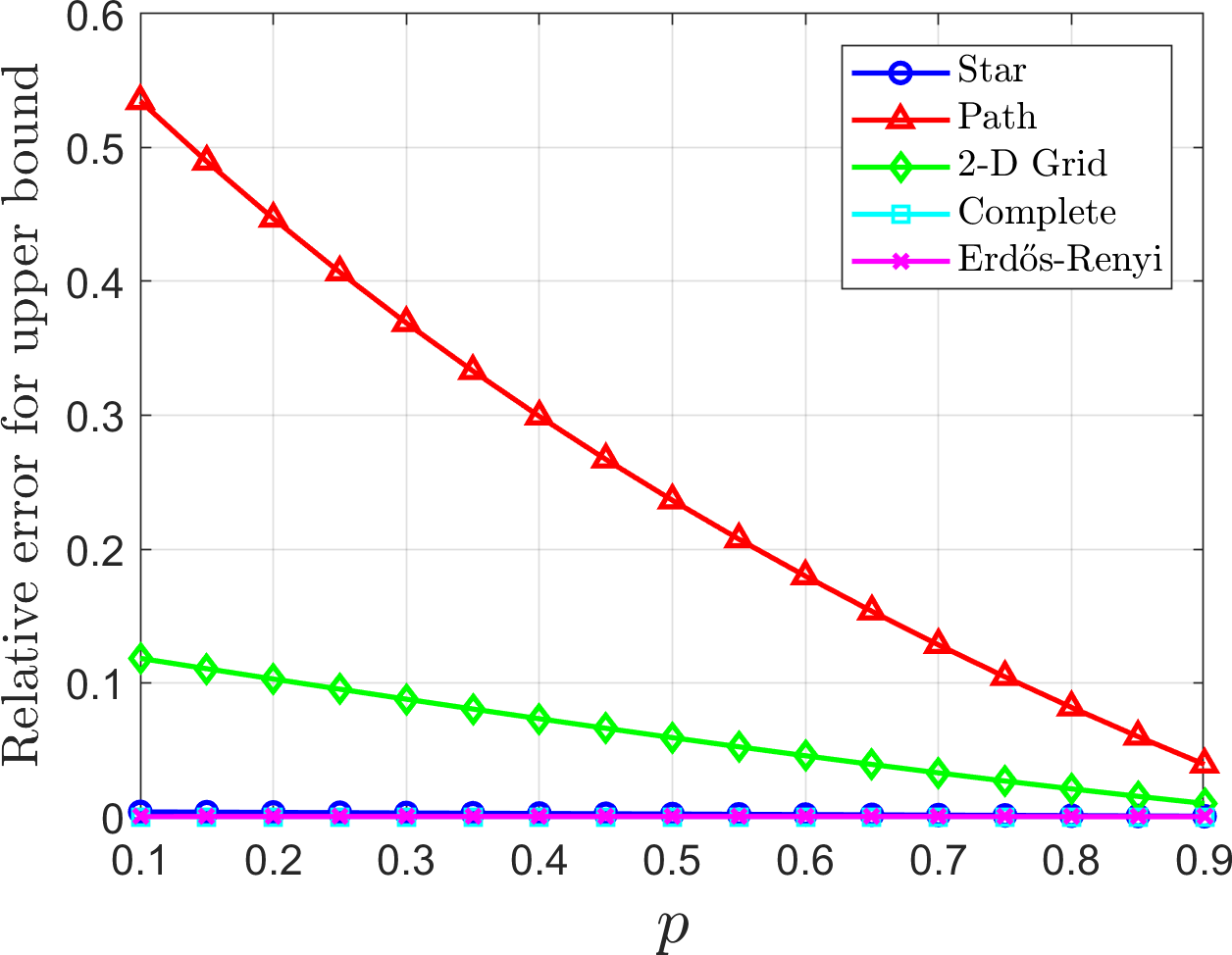}\label{fig:upper_bound_p}
}
\caption{Computations of the relative error for different values of $p$ for graphs with $N=100$.}
\end{figure}

\section{Conclusions}
In this work, we analyzed the influence of noise in random consensus with an application to OMAS by considering the random sampling of active nodes. We derived an expression for a noise index that measures the deviation from the consensus point and two bounds that depend on the spectrum of expected matrices related to the network. For Randomly Induced Discretized Laplacians, the bounds can be expressed as functions of the eigenvalues of the expected Laplacian matrix and, more specifically, of its average effective resistance. 
For sparse graphs, such as stars and paths, 
the noise index grows linearly and the lower bound presents a better approximation of the real value. For dense graphs, such as complete and Erd\H{o}s-R\'{e}nyi graphs, 
the noise index is bounded in $N$ and the upper bound is a better approximation. 

As future work, we would like to perform the analysis of noisy random consensus while allowing for a certain degree of statistical dependence in the generation of the stochastic matrices: such an extension would be very useful to study OMAS under more realistic assumptions.

\renewcommand{\newblock}{}
\bibliographystyle{IEEEtran}
\bibliography{references}

\begin{thebibliography}{10}
\providecommand{\url}[1]{#1}
\csname url@samestyle\endcsname
\providecommand{\newblock}{\relax}
\providecommand{\bibinfo}[2]{#2}
\providecommand{\BIBentrySTDinterwordspacing}{\spaceskip=0pt\relax}
\providecommand{\BIBentryALTinterwordstretchfactor}{4}
\providecommand{\BIBentryALTinterwordspacing}{\spaceskip=\fontdimen2\font plus
\BIBentryALTinterwordstretchfactor\fontdimen3\font minus
  \fontdimen4\font\relax}
\providecommand{\BIBforeignlanguage}[2]{{%
\expandafter\ifx\csname l@#1\endcsname\relax
\typeout{** WARNING: IEEEtran.bst: No hyphenation pattern has been}%
\typeout{** loaded for the language `#1'. Using the pattern for}%
\typeout{** the default language instead.}%
\else
\language=\csname l@#1\endcsname
\fi
#2}}
\providecommand{\BIBdecl}{\relax}
\BIBdecl

\bibitem{tahbaz2009consensus}
A.~Tahbaz-Salehi and A.~Jadbabaie, ``Consensus over ergodic stationary graph
  processes,'' \emph{IEEE Transactions on Automatic Control}, vol.~55, no.~1,
  pp. 225--230, 2009.

\bibitem{ravazzi2015ergodic}
C.~{Ravazzi}, P.~{Frasca}, R.~{Tempo}, and H.~{Ishii}, ``Ergodic randomized
  algorithms and dynamics over networks,'' \emph{IEEE Transactions on Control
  of Network Systems}, vol.~2, no.~1, pp. 78--87, 2015.

\bibitem{acemouglu2013opinion}
D.~Acemo{\u{g}}lu, G.~Como, F.~Fagnani, and A.~Ozdaglar, ``Opinion fluctuations
  and disagreement in social networks,'' \emph{Mathematics of Operations
  Research}, vol.~38, no.~1, pp. 1--27, 2013.

\bibitem{kar2008sensor}
S.~Kar and J.~M. Moura, ``Sensor networks with random links: Topology design
  for distributed consensus,'' \emph{IEEE Transactions on Signal Processing},
  vol.~56, no.~7, pp. 3315--3326, 2008.

\bibitem{bolognani2015randomized}
S.~Bolognani, R.~Carli, E.~Lovisari, and S.~Zampieri, ``A randomized linear
  algorithm for clock synchronization in multi-agent systems,'' \emph{IEEE
  Transactions on Automatic Control}, vol.~61, no.~7, pp. 1711--1726, 2015.

\bibitem{fagnani2008randomized}
F.~Fagnani and S.~Zampieri, ``Randomized consensus algorithms over large scale
  networks,'' \emph{IEEE Journal on Selected Areas in Communications}, vol.~26,
  no.~4, pp. 634--649, 2008.

\bibitem{xiao2007distributed}
L.~Xiao, S.~Boyd, and S.-J. Kim, ``Distributed average consensus with
  least-mean-square deviation,'' \emph{Journal of parallel and distributed
  computing}, vol.~67, no.~1, pp. 33--46, 2007.

\bibitem{jadbabaie2018scaling}
A.~Jadbabaie and A.~Olshevsky, ``Scaling laws for consensus protocols subject
  to noise,'' \emph{IEEE Transactions on Automatic Control}, vol.~64, no.~4,
  pp. 1389--1402, 2018.

\bibitem{yaziciouglu2019structural}
Y.~Yaz{\i}c{\i}o{\u{g}}lu, W.~Abbas, and M.~Shabbir, ``Structural robustness to
  noise in consensus networks: Impact of average degrees and average
  distances,'' in \emph{2019 IEEE 58th Conference on Decision and Control
  (CDC)}.\hskip 1em plus 0.5em minus 0.4em\relax IEEE, 2019, pp. 5444--5449.

\bibitem{garin2010survey}
F.~Garin and L.~Schenato, ``A survey on distributed estimation and control
  applications using linear consensus algorithms,'' in \emph{Networked control
  systems}.\hskip 1em plus 0.5em minus 0.4em\relax Springer, 2010, pp. 75--107.

\bibitem{patterson2010leader}
S.~Patterson and B.~Bamieh, ``Leader selection for optimal network coherence,''
  in \emph{49th IEEE Conference on Decision and Control (CDC)}.\hskip 1em plus
  0.5em minus 0.4em\relax IEEE, 2010, pp. 2692--2697.

\bibitem{bamieh2012coherence}
B.~Bamieh, M.~R. Jovanovic, P.~Mitra, and S.~Patterson, ``Coherence in
  large-scale networks: Dimension-dependent limitations of local feedback,''
  \emph{IEEE Transactions on Automatic Control}, vol.~57, no.~9, pp.
  2235--2249, 2012.

\bibitem{wang2012distributed}
J.~Wang and N.~Elia, ``Distributed averaging under constraints on information
  exchange: emergence of {L}{\'e}vy flights,'' \emph{IEEE Transactions on
  Automatic Control}, vol.~57, no.~10, pp. 2435--2449, 2012.

\bibitem{hendrickx2016open}
J.~M. Hendrickx and S.~Martin, ``Open multi-agent systems: Gossiping with
  deterministic arrivals and departures,'' in \emph{2016 54th Annual Allerton
  Conference on Communication, Control, and Computing (Allerton)}.\hskip 1em
  plus 0.5em minus 0.4em\relax IEEE, 2016, pp. 1094--1101.

\bibitem{hendrickx2017open}
------, ``Open multi-agent systems: Gossiping with random arrivals and
  departures,'' in \emph{2017 IEEE 56th Annual Conference on Decision and
  Control (CDC)}.\hskip 1em plus 0.5em minus 0.4em\relax IEEE, 2017, pp.
  763--768.

\bibitem{franceschelli2018proportional}
M.~Franceschelli and P.~Frasca, ``Proportional dynamic consensus in open
  multi-agent systems,'' in \emph{2018 IEEE Conference on Decision and Control
  (CDC)}.\hskip 1em plus 0.5em minus 0.4em\relax IEEE, 2018, pp. 900--905.

\bibitem{varma2018open}
V.~S. Varma, I.-C. Mor{\u{a}}rescu, and D.~Ne{\v{s}}i{\'c}, ``Open multi-agent
  systems with discrete states and stochastic interactions,'' \emph{IEEE
  Control Systems Letters}, vol.~2, no.~3, pp. 375--380, 2018.

\bibitem{fagnani2017introduction}
F.~Fagnani and P.~Frasca, \emph{Introduction to averaging dynamics over
  networks}.\hskip 1em plus 0.5em minus 0.4em\relax Springer, 2017.

\bibitem{del2011majorization}
S.~Del~Favero and S.~Zampieri, ``A majorization inequality and its application
  to distributed {Kalman} filtering,'' \emph{Automatica}, vol.~47, no.~11, pp.
  2438--2443, 2011.

\bibitem{lovisari2013resistance}
E.~Lovisari, F.~Garin, and S.~Zampieri, ``Resistance-based performance analysis
  of the consensus algorithm over geometric graphs,'' \emph{SIAM Journal on
  Control and Optimization}, vol.~51, no.~5, pp. 3918--3945, 2013.

\end{thebibliography}

\end{document}